\documentclass[12pt]{article}
\usepackage[margin=2.9cm]{geometry}
\usepackage{algorithmicx}

\usepackage[T1]{fontenc}
\usepackage[utf8]{inputenc}
\usepackage[english]{babel}
\usepackage{amsmath,amssymb, amsthm}
\usepackage[bookmarksopen,colorlinks,linkcolor=blue]{hyperref}
\usepackage{algorithm}
\usepackage{algpseudocode}
\usepackage{authblk}
\usepackage[shortlabels]{enumitem}
\usepackage{graphicx}
\usepackage{enumitem}

\usepackage{lmodern}
\usepackage{color}

\newcommand{\source}{\mathsf{source}}
\newcommand{\target}{\mathsf{target}}

\newcommand{\coun}{\mathsf{count}}

\title{State Complexity of Chromatic Memory in Infinite-Duration Games} 
\author{Alexander Kozachinskiy\footnote{kozlach@mail.ru}}

\newtheorem{theorem}{Theorem}
\newtheorem{corollary}[theorem]{Corollary}
\newtheorem{lemma}[theorem]{Lemma}

\newtheorem{definition}{Definition}

\newcommand{\col}{\mathsf{col}}
\newcommand{\mm}{\mathsf{mm}}

\makeatletter
\newcommand*{\rom}[1]{\expandafter\@slowromancap\romannumeral #1@}
\makeatother

\begin{document}
\maketitle

\begin{abstract}
A major open problem in the area of infinite-duration games is to characterize winning conditions that are determined in finite-memory strategies. Infinite-duration games are usually studied over edge-colored graphs, with winning conditions that are defined in terms of sequences of colors.
In this paper, we investigate a restricted class of finite-memory strategies called \emph{chromatic} finite-memory strategies. While general finite-memory strategies operate with sequences of edges of a game graph, chromatic finite-memory strategies observe only colors of these edges.

Recent results in this area show that studying  finite-memory determinacy is more tractable when we restrict ourselves to chromatic strategies.
On the other hand, 
as was shown by Le Roux (CiE 2020), determinacy in general finite-memory strategies implies determinacy in chromatic finite-memory strategies. Unfortunately, this result is quite inefficient in terms of the state complexity: to replace a winning strategy with few states of general memory, we might need much more states of chromatic memory. The goal of the present paper is to find out the exact state complexity of this transformation.

For every winning condition and for every game graph with $n$ nodes we show the following: if this game graph has a winning strategy with $q$ states of general memory, then it also has a winning strategy with $(q + 1)^n$ states of chromatic memory. We also show that this bound is almost tight. For every $q$ and $n$, we construct a winning condition and a game graph with $n + O(1)$ nodes, where one can win with $q$ states of general memory, but not with $q^n - 1$ states of chromatic memory.

\end{abstract}

\section{Introduction}

We study deterministic, two-player, turn-based, infinite-duration games on finite arenas. An arena is a finite directed graph whose set of nodes $V$ is partitioned into two subsets $V_0, V_1\subseteq V$. A game over such an arena is played as follows. Two players (in this paper, we call them Player 0 and Player 1) are traveling over the nodes of the graph. When they are in a node from $V_0$, Player $0$ decides where to go next. Similarly, when they are in a node from $V_1$, Player 1 decides where to go next. Players can move only along edges of the graph.

After infinitely many turns, this process produces an infinite path in our graph. A \emph{winning condition} determines, which infinite paths are winning for Player 0 and which are for Player 1. In the literature, one usually defines winning conditions through \emph{colors}. Namely, fix a set $C$ of colors. Then color edges of an arena into elements of $C$. Once this is done, any subset $W\subseteq C^\omega$ defines a winning condition in the following way. Take an infinite path in our graph. If its ``coloring'' (that is, the sequence of colors of its edges) belongs to $W$, we make this path winning for Player 0. Otherwise, we make this path winning for Player 1. In this paper, we consider only winning conditions that are defined through colors.

A utility of colors is that we do not have to define a winning condition for each arena separately. Indeed, once we fix a set of colors $C$ and a subset $W\subseteq C^\omega$, we can treat $W$ as a winning condition in any arena with edges colored by elements of $C$.

\medskip

A major open problem in this area is to characterize winning conditions that are determined in \emph{finite-memory strategies}. This is motivated by applications to reactive synthesis~\cite{thomas2004automata},  where infinite-duration games on graphs model an interaction between a system and the environment. Finite-memory determinacy over \emph{infinite} arenas has also been used in decidability of logical theories~\cite{muchnik1992game}.

Let us discuss the concept of a finite-memory strategy in more detail. The games we are studying are of infinite duration. This means that as a game goes on, the players need more and more memory to remember everything what has happened so far. However, a player might have a smart strategy which does not require the whole ``transcript''  of the previous development of the game. For instance, a strategy might only require to know whether the number of steps made so far is even or odd.
 In this case, 1 bit of memory is sufficient to ``implement'' this strategy (we simply reverse the value of this bit each time we make a step). In general, if there exists a constant $k > 0$ such that our strategy never stores more than $k$ bits of information during the play, then we call this strategy finite-memory.

More formally, finite-memory strategies store information through \emph{memory structures}.
A memory structure $\mathcal{M}$ is a finite automaton whose input alphabet is the set of edges of an arena. We call a strategy of one of the players an \emph{$\mathcal{M}$-strategy} if this strategy stores information according to $\mathcal{M}$. More specifically, one can imagine that, say, Player 0  carries $\mathcal{M}$ with them during the play. In each turn, Player 0 takes the edge along which the game goes on in this turn, and feeds this edge into $\mathcal{M}$. Then, in each moment of the game, the current state of $\mathcal{M}$ informally represents the memory of Player 0. If a strategy $S$ of this player never requires anything besides the current state of $\mathcal{M}$, then we call $S$ an $\mathcal{M}$-strategy. 

We call $S$ a \emph{finite-memory strategy} if it is an $\mathcal{M}$-strategy for some memory structure $\mathcal{M}$. If a memory structure $\mathcal{M}$ has $q$ states, then we call any $\mathcal{M}$-strategy a $q$-state strategy.

We say that a winning condition $W\subseteq C^\omega$ is \emph{finite-memory determined} for Player $i\in\{0, 1\}$ if in any arena, where this player has a winning strategy w.r.t.~$W$, this player also has a finite-memory winning strategy. In practice, we are interested in finite-memory winning strategies with as few states as possible. This leads to a concept of \emph{memory complexity}. To define it, we first have to fix a parametrization of arenas by natural numbers. For every $n\in\mathbb{N}$, there has to be only finitely many arenas where the value of the parameter does not exceed $n$.  When the set of colors is finite, one natural example of such a parameter is the number of nodes. Now,
the memory complexity of a finite-memory determined winning condition $W$ for Player $i\in\{0,1\}$ is the following function. For every $n\in\mathbb{N}$, it maps $n$ to smallest $q$ such that in all arenas, where Player $i$ can win w.r.t.~$W$ and where the value of the parameter does not exceed $n$, Player $i$ has a $q$-state winning strategy.

\medskip

There has been the following progress in the problem of characterizing  finite-memory  determined winning conditions. In 2005, Gimbert and Zielonka~\cite{gimbert2005games} obtained a sufficient and necessary condition for determinacy in \emph{memory-less strategies} (that is, in strategies that are defined over ``useless'' 1-state memory structures). Since then, several conditions that are sufficient (but not necessary) for finite-memory determinacy were established in the literature~\cite{le2018extending,bouyer2021finite}.

Recently, there have been advances in studying \emph{chromatic} memory complexity. In contrast to the general memory complexity, the chromatic memory complexity takes into account only so-called chromatic strategies (we define this kind of strategies in the next subsection). Bouyer et al.~\cite{bouyer_et_al:LIPIcs:2020:12836} obtained a complete characterization of winning conditions (and, more generally, of preference relations) with constant chromatic memory complexity. It is open to extend this result to winning conditions with constant\footnote{As we mentioned, memory complexity depends on the parametrization of the arenas. However, the class of winning conditions with constant memory complexity is independent of the choice of the parametrization. Indeed, these are winning conditions for which there exists a constant $q$ such that all arenas require at most $q$ states of memory.} general memory complexity. This motivates a study of the relationship between chromatic and general memory complexities -- as, potentially, there might be a way of converting results about chromatic memory into results about general memory. The present paper is devoted to this kind of questions.

\subsection{Chromatic strategies}

 A chromatic strategy is a strategy which, to make its moves, only needs to know the sequence of colors along a current play. That is, chromatic strategies use only a part of the description of a current play; the full description consists of  the sequence of edges that were played so far.

 We are interested in chromatic strategies that are additionally finite-memory. Such strategies can be defined through chromatic memory structures. A memory structure is chromatic if its transition function does not distinguish edges of the same color. Alternatively, a chromatic memory structure can be viewed as a finite automaton whose input alphabet is not the set of edges, but rather the set of colors. That is, each time we feed an edge into a chromatic memory structure, it only takes as an input the color of this edge.

In this paper, we study the following problem. Fix a winning condition and an arena. Suppose that Player $i$ has a winning strategy with $q$ states of general memory in this arena. What is the minimal $Q$ such that Player $i$ has a winning strategy with $Q$ states of chromatic memory in this arena?

It is not even obvious whether $Q$ is always finite (in other words, whether determinacy in general finite-memory strategies implies determinacy in chromatic finite-memory strategies). The finiteness of $Q$ was proved by Le Roux~\cite{le2020time}. In this paper, we obtain a tight bound on $Q$ in terms of $q$ and the number of nodes of an arena. But before presenting this in more detail, let us give an example where $Q$ is strictly larger than $q$.

\subsection{Example}
\label{subsec:example}

\begin{figure}[h!]
\centering
  \includegraphics[width=0.5\textwidth]{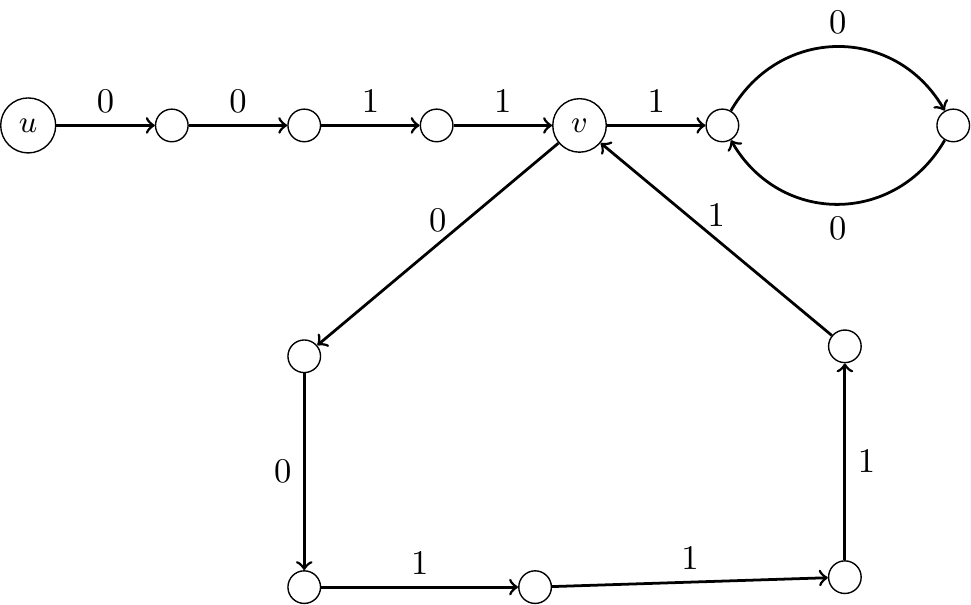}
  \caption{Example.}
\label{chr_example}
\end{figure}

In this example, the set of colors is $\{0,1\}$, the starting node is $u$, and Player 0 is the one to move everywhere. 
 Assume that the goal of Player 0 is to maximize the maximum number of consecutive 1's in the play. The best number Player 0 can achieve is 5. Namely, Player 0 can go to $v$, then to the loop and, after returning to $v$, to the right. Note that we do different things at $v$ when we are there for the first time and when for the second time. That is, the content of our memory has to be different at these moments. This can be realized with a 2-state memory structure which changes its state from ``has not been to $v$''  to ``has been to $v$'' after seeing any edge from $v$. However, such a memory structure is not chromatic -- it does different things for edges of the same color.

In fact, one can show that to attain 5 in this problem, we require $Q\ge 3$ states  of chromatic memory. Indeed, when we first come to $v$, we have to go to the loop -- otherwise the maximum number of consecutive 1's will be 3. Now, when we return to $v$, our memory has to be different from what it was the first time -- otherwise we will stay on the loop forever, and the maximum number of consecutive 1's will be 4. Therefore, our chromatic memory structure has to come into different states on  $0011$ and on $0011001111$. It can be checked via a computer search that no 2-state deterministic automaton over the alphabet $\{0,1\}$ can distinguish these two words.

\subsection{Our contribution}

As we mentioned, Le Roux has shown~\cite{le2020time} that determinacy in general finite-memory strategies implies determinacy in chromatic finite-memory strategies. More specifically, Le Roux obtained the following bound. Fix an arena $\mathcal{A}$ with $n$ nodes and a winning condition $W$.  Assume that  Player 0 has in $\mathcal{A}$ a $q$-state strategy which is winning w.r.t.~$W$. Then Player 0 also has in $\mathcal{A}$ a $2^{q (n^2 + 1)}$-state chromatic strategy which is winning w.r.t.~$W$. 

In fact, Le Roux obtained this as a corollary of a more general result. Namely, he studied a problem of the uniformization of winning strategies. In this problem, given an equivalence relation on the set of finite paths of the arena, we want a winning strategy which does the same thing on paths from the same equivalence class. In our case, the equivalence relation is ``to have the same sequence of colors and the same endpoint''. Let us also mention that Le Roux obtained this result for a more general class of games than we consider in this paper -- for deterministic \emph{concurrent} games. In these games, in each turn two players make two moves \emph{simultaneously}, and then one computes the next node via some predetermined transition function.

\textbf{The summary of our results.} First, we obtain an improvement of the aforementioned result of Le Roux. Namely, for any winning condition $W$ and for any arena $\mathcal{A}$ with $n$ nodes we show the following: if Player $0$ has in $\mathcal{A}$ a $q$-state winning strategy w.r.t.~$W$, then Player 0 also has in $\mathcal{A}$ a $(q+1)^n$-state chromatic winning strategy w.r.t~$W$.
Of course, the same  result holds for Player 1. We also provide an analog of this result for preference relations, with slightly worse bound $Q \le (qn + 1)^n$. Finally, we show that our upper bound is essentially tight. Namely, for every $n$ and $q$ we provide a winning condition $W$ and an arena $\mathcal{A}$ with $n + 3$ nodes, where Player 0 has a winning $q$-state strategy w.r.t.~$W$, but no winning chromatic strategy with less than $q^n$ states.

In fact, our upper bounds hold for deterministic concurrent games as well. However, to simplify the exposition, we prove them only for turn-based games.

\medskip

What consequences does it have for the memory complexity?
The bound of Le Roux means that the chromatic memory complexity is at most exponential in the general memory complexity and in the number of nodes. Our improvement is a removal of the exponential dependency on the general memory complexity. Still, there is an exponential dependency on the number of nodes. Note that this exponential dependency turns into just a linear dependency if we measure memory not in states but in bits. As a consequence, we obtain that for a class ``PSPACE'' of winning conditions with polynomial bit-memory complexity (in the number of nodes), it is unimportant whether we use chromatic or general strategies.

Unfortunately, our lower bound does not provide any separation between the chromatic and the general memory complexity. We only provide an exponential separation between the chromatic and the general memory in individual arenas, using different winning conditions for different arenas. It is not clear how to combine these winning conditions into a single one for which the general memory grows more slowly than the chromatic memory. We return to this problem in the next subsection.

\subsection{Other related works and open problems}

 Chromatic finite-memory strategies were first studied by Kopczy\'{n}ski~\cite{phdthesis} in 2008. He introduced the following notation for winning conditions with constant memory complexity. For $i\in\{0,1\}$ and for a winning condition $W\subseteq C^\omega$, let $\mm_i(W)$ denote the minimal $q$ such that for every finite arena the following holds: if Player $i$ has a winning strategy w.r.t.~$W$ in this arena, then Player $i$ has also a $q$-state winning strategy w.r.t.~$W$ in this arena. If no such $q$ exists, then we set $\mm_i(W)=+\infty$. That is, $W$ has constant memory complexity for Player $i$ if and only if $\mm_i(W)$ is finite.  Kopczy\'{n}ski also introduced an analogue of this notation for chromatic strategies.
Namely, let $\mm_i^\chi(W)$ 
denote the minimal $q$ such that for every finite arena the following holds: if Player $i$ has a winning strategy w.r.t.~$W$ in this arena, then Player $i$ has also a chromatic $q$-state winning strategy w.r.t.~$W$ in this arena. Again, if no such $q$ exists, we set $\mm_i^\chi(W) = +\infty$. Clearly, we have $\mm_i(W) \le \mm^\chi_i(W)$ for every winning condition $W$.

It is classical~\cite{buchi1969solving} that all $\omega$-regular winning conditions $W\subseteq C^\omega$ have constant chromatic memory complexity: $\mm_0^\chi(W) < +\infty, \mm_1^\chi(W)< +\infty$. Kopczy\'{n}ski gave an exponential-time algorithm which computes $\mm_i^\chi(W)$ for an $\omega$-regular $W$, provided that $W$ is given as a  deterministic parity automaton. He then conjectured that $\mm_i^\chi(W) = \mm_i(W)$ for every winning condition $W$. His motivation was as follows: if this equality holds, then we automatically have an exponential-time algorithm to compute $\mm_i(W)$ for an $\omega$-regular $W$.

As we mentioned, in 2020 Bouyer et al.~\cite{bouyer_et_al:LIPIcs:2020:12836} obtained a \emph{complete characterization} of winning conditions with constant chromatic memory complexity for both players, that is, of winning conditions satisfying $\mm^\chi_0(W) < +\infty, \mm^\chi_1(W) <+\infty$. Broadly speaking, they gave a graph-theoretic sufficient and necessary condition for  determinacy with constant chromatic memory complexity. In other words, they reduced a game-theoretic reasoning in this case to a reasoning about graphs.

This result motivates the following modification of the Kopczy\'{n}ski's conjecture. Is it true for every winning condition $W$ that $\mm_i(W) < +\infty \iff \mm_i^\chi(W) < +\infty$? In other words, is it true that a winning condition with constant general memory complexity also has constant chromatic memory complexity? If that were true, we would have had a complete characterization of winning conditions with constant general memory complexity, due to results of Bouyer et al.

Recently, Casares~\cite{casares2021minimisation} disproved the original conjecture of Kopczy\'{n}ski. More specifically, for every $n$ he constructed a winning condition $W_n\subseteq \{1, 2, \ldots, n\}^\omega$ such that $\mm_0^\chi(W_n) = n$ and $\mm_0(W_n) = 2$. This winning condition is arranged as follows. We put $\alpha\in\{1, 2, \ldots, n\}^\omega$ into $W_n$ if and only if there are exactly two  elements of $\{1, 2, \ldots, n\}$ that occur infinitely often in $\alpha$. Let us note that $W_n$ is $\omega$-regular.

One can notice that this result of Casares also gives an example of a winning condition $W$ with constant general memory complexity, but super-constant chromatic memory complexity. Thus, it answers negatively to our modification of the Kopczy\'{n}ski's conjecture.  However, this $W$ is over an \emph{infinite} set of colors. Namely, set $C = \mathbb{N}$. Let $W\subseteq \mathbb{N}^\omega$ be the set of all $\alpha\in\mathbb{N}^\omega$ such that there are exactly 2 elements of $\mathbb{N}$ occurring infinitely often in $\alpha$. On one hand, $W$ is at least as hard (in terms of memory requirements) as $W_n$, for every $n$, so we have $\mm_0^\chi(W) = +\infty$. On the other hand, since any finite arena involves only finitely many colors, and since $\mm_0(W_n) = 2$ for every $n$, we have $\mm_0(W) = 2$.

It is open if there exists a winning condition $W$ over a \emph{finite} set of colors with $\mm_0(W) < +\infty$ and $\mm^\chi_0(W) = +\infty$. Let us note here that such a separation can be easily obtained for \emph{one-player} games (that is, for games over arenas where one of the players controls all the nodes), via the example from subsection \ref{subsec:example}. Namely, let the set of colors be $\{0, 1\}$, and assume that the goal of Player 0 is to maximize the maximum number of consecutive 1's in a sequence of colors. It is not hard to see that Player 0 can play optimally with just 2 states of memory in every one-player arena. An idea is that we first have to reach the starting point of the longest all-ones path $p$, and then we have to go through $p$. We may assume that $p$ is simple because otherwise there is a simple all-ones cycle which allows us to win trivially. There is a slight problem that the path to the starting node of $p$ and the path $p$ itself may have common nodes, and we have to do different things at these nodes at different moments. So we just have to remember whether we have already been to the starting node of $p$ or not. However, this kind of memory is essentially non-chromatic. In fact, by considering arenas as on Figure \ref{chr_example}, one can show that there is no constant $q$ such that Player 0 can play optimally with $q$ states of chromatic memory in all one-player arenas. Unfortunately, this example has super-constant general memory complexity in two-player arenas, so it does not provide a separation between the chromatic and the general memory complexity in the two-player regime.

The question of 
the relationship between the chromatic and the general memory is not settled in infinite arenas. In contrast to finite arenas,  it is unclear for infinite arenas whether determinacy in general finite-memory strategies implies determinacy in chromatic finite-memory strategies. Note that our upper bounds are inapplicable as they involve the number of nodes. It was recently shown by Bouyer, Randour and Vandenhove~\cite{bouyer2021characterizing} that the class of winning condition that are chromatic finite-memory determined in infinite arenas coincides with the class $\omega$-regular conditions.

Finally, it would be interesting to extend our upper bounds on chromatic memory to \emph{stochastic} games. Seemingly, our proofs do not work in the stochastic setting, so we require new ideas here. It is worth to mention in this context a paper~\cite{bouyer_et_al:LIPIcs.CONCUR.2021.26}, studying winning conditions that have constant chromatic memory complexity in stochastic games.

\medskip

The rest of the paper is organized as follows. In the next section we give preliminaries. In Section \ref{sec:exact} we give the exact statements of our results, equipped with brief overviews of the proofs. Complete proofs of our results are given in the subsequent sections.

\section{Preliminaries}
\textbf{Notation.} For a set $A$, we let $A^*$ (resp., $A^\omega$) stand for the set of all finite (resp., infinite) sequences of elements of $A$. For $x\in A^*$, we let $|x|$ denote the length of $x$ (we also set $|x| = +\infty$ for $x\in A^\omega$). We write $A = B\sqcup C$ for three sets $A, B, C$ when $A = B\cup C$ and $B\cap C = \varnothing$. We let $\circ$ denote the function composition. The set of positive integral numbers is denoted by $\mathbb{Z}^+$.
\subsection{Arenas}

\begin{definition} Let $C$ be any set.
A tuple $\mathcal{A} = \langle V, V_0, V_1, E\rangle$, where $V, V_0, V_1, E$ are four finite sets such that $V = V_0\sqcup V_1$ and $E\subseteq V\times C\times V$, is called an \textbf{arena} over the set of colors $C$ if for every $s\in V$ there exist $c\in C$ and $t\in V$ such that $(s, c, t)\in E$.
\end{definition}

Elements of $V$ will be called nodes of $\mathcal{A}$. Elements of $V_0$ (resp., $V_1$) will be called nodes of Player 0 (resp., Player 1). Elements of $E$ will be called edges of $\mathcal{A}$. For an edge $e = (s, c, t) \in E$ we define $\source(e) = s, \col(e) = c$ and $\target(e) = t$.
We imagine $e\in E$ as an arrow which is drawn from the node $\source(e)$ to the node $\target(e)$ and which is colored by $\col(e)$. The out-degree of a node $s\in V$ is the number of $e\in E$ with $s = \source(e)$. By definition of an arena, the out-degree of every node must be positive.

 We extend the function $\col$ to a function $\col\colon E^*\cup E^\omega\to C^*\cup C^\omega$ by setting:
\[\col(e_1 e_2 e_3\ldots) = \col(e_1)\col(e_2)\col(e_3)\ldots, \qquad e_1, e_2, e_3,\ldots\in E.\]

A non-empty sequence $p = e_1 e_2 e_3 \ldots \in E^*\cup E^\omega$ is called a path if for any $1\le i < |p|$ we have $\target(e_i) = \source(e_{i+1})$. We set $\source(p) = \source(e_1)$ and, if $p$ is finite, $\target(p) = \target(e_{|p|})$.
For technical convenience, with every node $v\in V$ we associate a $0$-length path $\lambda_v$, for which we set $\source(\lambda_v) = \target(\lambda_v) = v$. 

\subsection{Strategies}

Let $\mathcal{A} = \langle V, V_0, V_1, E\rangle$ be an arena over the set of colors $C$. A strategy of Player 0 in $\mathcal{A}$ is any function 
\[S\colon\{p\mid p \mbox{ is a finite path in $\mathcal{A}$ with }\target(p) \in V_0\}\to E,\]
such that for every $p$ from the domain of $S$ we have $\source(S(p)) = \target(p)$.

Intuitively, finite paths in $\mathcal{A}$ represent possible developments of the game in $\mathcal{A}$. When the game starts in a node $s\in V$, the starting position\footnote{We stress that in our setting, the starting node is not fixed. The same strategy $S$ can be applied for different starting nodes.} is represented by the $0$-length path $\lambda_s$. Player $0$ is the one to move after a finite path $p$ if and only if $t =\target(p)$ is a node of Player 0. In this situation, Player 0 must choose some edge from $t$. We then append this edge to $p$. A strategy of Player $0$ fixes the choices of Player 0 in all situations when this player is the one to move.

In this paper, we only mention strategies of Player 0, but, of course, strategies of Player 1 can be defined similarly.

Let us define, what does it mean that a path is consistent with a strategy $S$ of Player 0. First, any 0-length path $\lambda_v$ is consistent with $S$. Now, a non-empty  path $p = e_1 e_2 e_3\ldots$ (which may be finite or infinite) is consistent with  $S$ if the following holds:
\begin{itemize}
\item if $\source(p) \in V_0$, then $e_1 = S(\lambda_{\source(p)})$;
\item for any $1 \le i < |p|$, if $\target(e_i) \in V_0$, then $e_{i + 1} = S(e_1 e_2\ldots e_i)$.
\end{itemize}

For any node $v\in V$ and for any strategy $S$ of Player 0, we let $\col(S, v)\subseteq C^\omega$ be the set of all $\col(p)$ over all infinite paths $p$ such that $\source(p) = v$ and $p$ is consistent with $S$. Less formally, $\col(S, v)$ is the set of all infinite sequences of colors that can be obtained in a play with $S$ from the node $v$. For $U\subseteq V$, we define $\col(S, U) = \bigcup_{v\in U}\col(S, v)$.
\subsection{Winning conditions and preference relations}
A winning condition is any set $W\subseteq C^\omega$. We say that a strategy $S$ of Player 0 is winning from $u\in V$ w.r.t.~to $W$ if $\col(S, u)\subseteq W$. In other words, any infinite play from $u$ against $S$ must give a sequence of colors belonging to $W$.

We also consider a more general class of objectives called \emph{preference relations}.
A preference relation is a total preorder $\sqsubseteq$ on the set $C^\omega$.
Intuitively, when given a preference relation $\sqsubseteq$, the goal of Player 0 is to maximize the sequence of colors in the play w.r.t.~$\sqsubseteq$.

Any two strategies $S_1, S_2$ of Player 0 can be compared w.r.t.~$\sqsubseteq$ (from the perspective of Player 0). Namely, we say that $S_1$ is \emph{better} than $S_2$ from $u\in V$ if there exists $\beta\in \col(S_2, u)$ such that $\alpha \not\sqsubseteq \beta$ for every $\alpha \in \col(S_1, u)$. This means that  there exists a play against $S_2$ which is strictly worse from the viewpoint of Player 0 than any play against $S_1$. Correspondingly, we say that $S_2$ is \emph{at least as good} as $S_1$ from $u\in V$ if it is not true that $S_1$ is better than $S_2$ from $u$, that is, if for every $\beta\in\col(S_2, u)$ there exists $\alpha\in \col(S_1, u)$ such that $\alpha\sqsubseteq \beta$.

It is straightforward to check that for any $u\in V$, the relation ``at least as good from $u$'' is a total preorder on the set of strategies of Player 0.
\subsection{Memory structures}

Let $\mathcal{A} = \langle V, V_0, V_1, E\rangle$ be an arena over the set of colors $C$.
A memory structure in $\mathcal{A}$ is a tuple $\mathcal{M} = \langle M, m_{init}, \delta\rangle$, where $M$ is a finite set, $m_{init}\in M$ and $\delta\colon M\times E\to M$. In other words, a memory structure $\mathcal{M}$ is a deterministic finite automaton whose input alphabet is the set of edges of $\mathcal{A}$. Thus, $M$ serves as the set of states of our memory structure, $m_{init}$ serves as its initial state, and $\delta$ as its transition function. Given $m\in M$, we inductively extend the function $\delta(m,\cdot)$ to arbitrary finite sequences of edges as follows:
\begin{align*}
\delta(m, \mbox{empty sequence}) &= m,\\ 
\delta(m, se) &= \delta(\delta(m, s), e), \qquad s\in E^*, e\in E.
\end{align*}
Thus, $\delta(m, s)$ for $s\in E^*$ is the state into which our memory structure comes from the state $m$ after reading the sequence $s$.

We say that a memory structure $\mathcal{M} = \langle M, m_{init}, \delta\rangle$ is chromatic if there exists a function $\sigma\colon M \times C\to M$ such that $\delta(m, e) = \sigma(m, \col(e))$ for every $e\in E$ and $m\in M$.
In other words, a chromatic memory structure does not distinguish edges of the same color. Correspondingly, 
it will be sometimes convenient to view chromatic memory structures as finite automata over the set $C$ (and not over the set of edges of $\mathcal{A}$).

Let $S$ be a strategy of Player 0 and $\mathcal{M} = \langle M, m_{init}, \delta\rangle$ be a memory structure. We say that $S$ is an $\mathcal{M}$-strategy if for any two paths $p_1, p_2$ with $\target(p_1) = \target(p_2)\in V_0$ it holds that:
\[\delta(m_{init}, p_1) = \delta(m_{init}, p_2) \implies S(p_1) = S(p_2)\]
In other words, the value of $S(p)$ for an $\mathcal{M}$-strategy $S$ solely depends on $\target(p)$  and on a state of $\mathcal{M}$  after $p$ (assuming that initially $\mathcal{M}$ is in $m_{init}$). Thus, with any $\mathcal{M}$-strategy $S$ one can associate the \emph{next-move} function of $S$. This is a function $n_S\colon V_0\times M\to E$ defined as follows: to determine $n_S(v, m)$, we take an arbitrary finite path $p$ with $\target(p) = v$ and $\delta(m_{init}, p) = m$, and set $n_S(v, m) = S(p)$. If there is no such path $p$ at all, we define $n_S(v, m)$ arbitrarily. Less formally, $n_S(v, m)$ is the move of $S$ from the node $v$ when the state of $S$ is $m$. Note that the next-move function completely determines the corresponding strategy. For the sake of brevity, in the paper we will use the same letter for a strategy and for its next-move function. That is, when $S$ is an $\mathcal{M}$-strategy, we use the letter $S$ in two different ways. First, $S(p)$ denotes the move of $S$ after a finite path $p$. Second, $S(v, m)$ for $v\in V_0, m\in M$ denotes the value of the next-move function of $S$ on the pair $(v, m)$.

For brevity, if $S$ is an $\mathcal{M}$-strategy, we let \emph{the state of $S$ after a finite path $p$}  stand for $\delta(m_{init}, p)$, where $\delta$ is the transition function of $\mathcal{M}$.

We say that $S$ is a $q$-state strategy if $S$ is an $\mathcal{M}$-strategy for some memory structure $\mathcal{M} = \langle M, m_{init}, \delta\rangle$ with $|M| = q$. If $\mathcal{M}$ is additionally chromatic, then we say that $S$ is a chromatic $q$-state strategy.

\section{Exact Statements of the Results and Overviews of the Proofs}
\label{sec:exact}

The exact statement of our main upper bound is the following
\begin{theorem}
\label{thm:first_upper}
 For any $n, q\in\mathbb{Z}^+$, for any arena $\mathcal{A} = \langle V, V_0, V_1, E\rangle$ with $n$ nodes, for any set $U\subseteq V$, and for any $q$-state strategy $S_1$ of Player 0 in $\mathcal{A}$, there exists a chromatic $(q + 1)^n$-state strategy $S_2$ of Player 0 such that $\col(S_2, U) \subseteq \col(S_1, U)$.
\end{theorem}

It should not be confusing that this theorem does not mention winning conditions. One can notice that $\col(S_1, U)$ is the minimal winning condition w.r.t.~which $S_1$ is winning from all nodes of $U$. So, if $S_1$ and $S_2$ are as in Theorem \ref{thm:first_upper}, and $W$ is a winning condition w.r.t.~which $S_1$ is winning from all nodes of $U$, then $S_2$ is also winning w.r.t.~$W$ from all nodes of $U$. That is, we obtain the following corollary.

\begin{corollary}
\label{col:condition}
Let $W\subseteq C^\omega$ be any winning condition. Then for any $n, q\in\mathbb{Z}^+$ the following holds. Take any $n$-node arena $\mathcal{A}$ and any $q$-state strategy $S_1$ of Player 0 in $\mathcal{A}$. Then in $\mathcal{A}$ there exists a chromatic $(q + 1)^n$-state strategy $S_2$ of Player 0 such that for any node $v$ of $\mathcal{A}$ the following holds: if $S_1$ is winning from $v$ w.r.t.~$W$, then so is $S_2$.
\end{corollary}
\begin{proof}
Apply Theorem \ref{thm:first_upper} to the set $U$ of nodes from where $S_1$ is winning w.r.t.~$W$.
\end{proof}

Let us now give a proof sketch of Theorem \ref{thm:first_upper}. Its full proof is given in Section \ref{sec:first_upper}

\begin{proof}[Proof sketch of Theorem \ref{thm:first_upper}]
At each moment, $S_2$ stores a function $f\colon V\to M\cup\{\bot\}$, where $M$ is the set of states of $S_1$ and $\bot\notin M$ is a special symbol meaning ``undefined''. There are $(q + 1)^n$ such functions, so $S_2$ will have $(q + 1)^n$ states. Strategy $S_2$ maintains the following two invariants called \emph{soundness} and \emph{completeness}:
\begin{itemize}
    \item (\emph{soundness}) for any $v\in V$ where $f$ is defined, there exists a path from $U$ to $v$ such that, first, this path is colored exactly as our current play with $S_2$, second, this path is consistent with $S_1$, and third, the state of $S_1$ after this path is $f(v)$; 
    
    \item (\emph{completeness}) $f$ is defined in the last node of our current play with $S_2$.
\end{itemize}

To show that $\col(S_2, U)\subseteq \col(S_1, U)$, we have to show that for any infinite play $P_2$ with $S_2$ from $U$ there exists an infinite play $P_1$ with $S_1$ from $U$ such that $P_1$ has the same sequence of colors as $P_2$. An analog of this statement for finite plays easily follows from the soundness and the completeness. Indeed, if we have a finite play $p_2$ with $S_2$, then, by the completeness, $f$ is defined on its endpoint, so we can take a finite play $p_1$ with $S_1$ which establishes the soundness for the endpoint of $p_2$. We then extend this to infinite plays via the Kőnig's lemma, by considering all finite prefixes of $P_2$.

It remains to define the memory structure of $S_2$ (which has to be chromatic) and the next-move function of $S_2$ in such a way that the soundness and the completeness hold.

 We start with the next-move function of $S_2$. Assume that $S_2$ has to make a move from $v\in V$. If $f(v) \neq \bot$, then $S_2$ makes the same move as $S_1$ does from $v$ when its state is $f(v)$. If $f(v) = \bot$, then $S_2$ makes an arbitrary move (however, the completeness invariant ensures that this situation never occurs).
 
 We now define the memory structure of $S_2$. To ensure our invariants in the beginning, we let the initial state of $S_2$ be the function $f_{init}\colon V\to M\cup\{\bot\}$, which maps nodes from $U$ to $m_{init}$ (the initial state of $S_1$) and other nodes to $\bot$. Now, assume that the current state of $S_2$ is represented by a function $f\colon V\to M\cup\{\bot\}$, and then we receive an edge whose color is $c\in C$. Our current play against $S_2$ became one edge longer. We have to define a new state $\widehat{f}\colon V\to M\cup\{\bot\}$ of $S_2$ in some way which preserves the soundness and the completeness. Moreover, since $S_2$ has to be chromatic, $\widehat{f}$ has to be a function of $f$ and the color $c$. 
 
 On a high level, we consider all paths which establish the soundness for $f$, and then try all possible ways to extend them by $c$-colored edges without violating consistency with $S_1$. Of course, we do not store all these paths in our memory, but the knowledge of $f$ is sufficient. For example, assume that we have a $c$-colored edge from a node $u$ of Player 1 with $f(u)\neq \bot$ to a node $v$. Then we can safely define $\widehat{f}(v)$ to be the state into which the memory structure of $S_1$ transits from the state $f(u)$ after reading the edge $(u, c, v)$. Now, if $u$ is a node of Player 0, there is one more restriction: the edge $(u, c, v)$ has to be the move of $S_1$ from $u$ in the state $f(u)$. This is needed to preserve the consistency with $S_1$.
 
For the completeness, we have to ensure that there is a $c$-colored edge allowing us to define $\widehat{f}$ at the endpoint of our current play with $S_2$. It is not hard to see that one such edge is simply the last edge of this play. It is not a problem that we might actually use some other edge to define $
\widehat{f}$ at the last node of our current play -- we should only ensure that $\widehat{f}$ is defined there.

\end{proof}

To obtain an analog of Corollary \ref{col:condition} for preference relations, we establish the following result.

\begin{theorem} 
\label{thm:second_upper}
For any $n, q\in\mathbb{Z}^+$, for any arena $\mathcal{A} = \langle V, V_0, V_1, E\rangle$ with $n$ nodes, for any total preorder $\preceq$ on the set $V$ and for any $q$-state strategy $S_1$ of Player 0 in $\mathcal{A}$, there exists a chromatic $(qn + 1)^n$-state strategy $S_2$ of Player 0 such that for any $v\in V$ we have:
\[\col(S_2, v) \subseteq \bigcup\limits_{u \in V, v \preceq u} \col(S_1, u).\]
\end{theorem}

\begin{corollary}
\label{col:pref}
Let $\sqsubseteq$ be any preference relation on $C^\omega$. Then for any $n, q\in\mathbb{Z}^+$ the following holds. Take any $n$-node arena $\mathcal{A}$ and any $q$-state strategy $S_1$ of Player 0 in $\mathcal{A}$. Then there exists a chromatic $(qn + 1)^n$-state strategy $S_2$ of Player 0 such that, for any node $v$ of $\mathcal{A}$, we have that $S_2$ is at least as good as $S_1$ w.r.t.~$\sqsubseteq$ from $v$.
\end{corollary}
Let us discuss why do we need slightly more states in Corollary \ref{col:pref} than in Corollary \ref{col:condition}. The reason is that we want $S_2$ to be as good as $S_1$ from \emph{every} node of $\mathcal{A}$. When we were dealing with winning condition, we could just forget about the nodes where $S_1$ is not winning. Now there is a finer classification of the nodes, depending on  what Player 0 can achieve in these nodes w.r.t.~$\sqsubseteq$, and it is slightly harder to deal with this classification.

Let us now formally derive Corollary \ref{col:pref} from Theorem \ref{thm:second_upper}.
\begin{proof}[Proof of Corollary \ref{col:pref}]
We write $u\preceq v$ for two nodes $u, v$ of $\mathcal{A}$ if for any $\beta\in\col(S_1, v)$ there exists $\alpha\in\col(S_1, u)$ such that $\alpha\sqsubseteq \beta$. Let us verify that $\preceq$ is a total preorder on the set $V$ of nodes of $\mathcal{A}$. The transitivity of $\preceq$ follows from the transitivity $\sqsubseteq$. The reflexivity of $\preceq$ is obvious. Now we show the totality of $\preceq$, that is, we show that $u\not\preceq v\implies v\preceq u$. Since $u\not\preceq v$, there exists $\beta\in \col(S_1, v)$ such that $\alpha\not \sqsubseteq \beta$ for every $\alpha\in \col(S_1, u)$. By the totality of $\sqsubseteq$, we have $\beta \sqsubseteq \alpha$ for every $\alpha\in\col(S_1, u)$. This implies that $v\preceq u$.

We then apply Theorem \ref{thm:second_upper} to $\preceq$. Consider the resulting chromatic $(qn + 1)^n$-state strategy $S_2$.  We show, for every $v\in V$, that $S_2$ is at least as good as $S_1$ w.r.t.~$\sqsubseteq$ from $v$. That is, we show, for every $\beta\in \col(S_2, v)$, that there exists $\alpha\in \col(S_1, v)$ with $\alpha\sqsubseteq\beta$.
 By the conclusion of Theorem \ref{thm:second_upper}, we have that $\beta$ belongs to $\col(S_1, u)$ for some $v\preceq u$. By definition of $\preceq$, there exists some $\alpha\in \col(S_1, v)$ such that $\alpha \sqsubseteq\beta$, as required.
\end{proof}

Next we provide a proof sketch of Theorem \ref{thm:second_upper}, whose full proof is given in Section \ref{sec:second_upper}.
\begin{proof}[Proof sketch of Theorem \ref{thm:second_upper}]
 This time, $S_2$ will store slightly more information than in the proof of Theorem \ref{thm:first_upper}. Namely, at each moment, $S_2$ stores a function of the form $f\colon V\to V\times M \cup\{\bot\}$ (observe that this requires $(qn + 1)^n$ states). That is, whenever $f$ is defined at $v$, the value $f(v)$ is an ordered pair in which the first coordinate is a node and the second coordinate is a state of $S_2$. We introduce the first coordinate to strengthen the soundness condition: a path which establishes the soundness at $v$ must be from the first coordinate of $f(v)$. More precisely, the soundness condition is formulated as follows:

\begin{itemize}
    \item (\emph{soundness}) for any $v\in V$ where $f(v) = (u, m)\neq \bot$ is defined, there exists a path from $u$ to $v$ such that, first, this path is colored exactly as our current play with $S_2$, second, this path is consistent with $S_1$, and third, the state of $S_1$ after this path is $m$; 
\end{itemize}
To take into account the preorder $\preceq$, we modify the completeness condition as follows:
\begin{itemize}
\item (\emph{completeness}) $f$ is defined at the last node of our current play with $S_2$; moreover, the first coordinate of $f$ at this node is at least as large w.r.t.~$\preceq$ as the starting node of our play with $S_2$. 
\end{itemize}

These two conditions imply that
 for any infinite play $P_2$ with $S_2$ there exists an infinite play $P_1$ with $S_1$ such that, first, $\col(P_2) = \col(P_1)$, and second, $\source(P_2) \preceq \source(P_1)$ (again, we first derive this for finite plays, and then extend to infinite plays via the Kőnig's lemma). In turn, this obviously implies that 
 \[\col(S_2, v) \subseteq \bigcup\limits_{u \in V, v \preceq u} \col(S_1, u).\]
 
 It remains to define $S_2$ satisfying the soundness and the completeness. We give essentially the same
 definition as in the proof of Theorem \ref{thm:first_upper}. The only difference is that we have to additionally care about the first coordinates of $f$. For example, assume that our current state is $f$, and then we have to determine the new state $\widehat{f}$ at a node $v$. If there is more than one edge of the form $(u, c, v)$ allowing to define $\widehat{f}$ at $v$, we choose one maximizing the first coordinate of $f(u)$ w.r.t.~$\preceq$ (and we set the first coordinate of $\widehat{f}(v)$  to be equal to this maximum). This ensures that the value of the first coordinate of $f(v_{cur})$, where $v_{cur}$ is the last node in our current play with $S_2$, can only increase over time w.r.t.~$\preceq$. In particular, it is always at least as large as the starting node of our play, since in the beginning we set $f_{init}(v) = (v, m_{init})$ for every $v\in V$. This easily establishes our modified completeness condition, while the soundness requires almost no new argument.

\end{proof}

Finally, the exact statement of our lower bound showing the tightness of Theorem \ref{thm:first_upper} is the following:

\begin{theorem}
\label{thm:lower_bound}
For any $n, q\in\mathbb{Z}^+$ there exists an arena $\mathcal{A}$ with $n + 3$ nodes, a node $u$ of $\mathcal{A}$ and a $q$-state strategy $S_1$ of Player 0 in $\mathcal{A}$  such that for any chromatic $Q$-state  strategy $S_2$ of Player 0 it holds that $\col(S_2, u) \subseteq \col(S_1, u) \implies Q\ge q^n$.
\end{theorem}

Our argument has a connection to a work of Jirásková and Pighizzini  
\cite{jiraskova2011optimal} on \emph{self-verifying} automata. It turns out that from one of their results one can directly derive a weaker version of Theorem \ref{thm:lower_bound}. Namely, one can get an arena with $n + O(1)$ nodes and a 2-state strategy $S_1$ of Player 0 in it such that for some node $u$ of this arena the following holds: if $S_2$ is a chromatic $Q$-state strategy with $\col(S_2, u)\subseteq \col(S_1, u)$, then $Q = \Omega(3^{n/2})$. We show this derivation below in this section. Of course, when $S_1$ has 2 states, Theorem 
\ref{thm:lower_bound} gives a better bound $Q = \Omega(2^n)$, let alone that $q$ can be arbitrary in Theorem \ref{thm:lower_bound}). In fact, for the strong version of Theorem \ref{thm:lower_bound} it is not sufficient to use results of Jirásková and Pighizzini as a black box -- we have to slightly modify their construction. The full proof of Theorem \ref{thm:lower_bound} is given in Section \ref{sec:lower_bound}. 

\begin{proof}[Sketch of the proof of Theorem \ref{thm:lower_bound} (weak version)]

A self-verifying automaton is a non-deterministic finite automaton with a property that, for any input word $w$, exactly one of the following two statements holds:
\begin{itemize}
    \item there exists a computation which accepts $w$;
    \item there exists a computation  which rejects $w$.
\end{itemize}

For every $n\in\mathbb{N}$, Jirásková and Pighizzini construct a self-verifying automaton $A_n$ with $n$ states such that any deterministic automaton, recognizing the same language as $A_n$, has $\Omega(3^{n/2})$ states (they also show that this bound is tight). Using $A_n$, we construct an arena with $n + O(1)$ nodes and a 2-state strategy $S_1$ of Player 0 in it such that for some node $u$ of this arena the following holds: if $S_2$ is a chromatic $Q$-state strategy with $\col(S_2, u)\subseteq \col(S_1, u)$, then $Q = \Omega(3^{n/2})$. Namely, consider the transition graph of $A_n$; it can be viewed as an arena with edges colored by the input letters of $A_n$. Assume that Player 1 is the one to move everywhere in this transition graph. Now, add a node $t$ controlled by Player 0. Draw edges to $t$ from all accepting and rejecting states of $A_n$. Color all these edges into a single color $\#$. Finally, take two colors $c$ and $d$ that do not belong to the input alphabet of $A_n$, and draw two edges from $t$ to the initial state of $A_n$, one colored by $c$ and the other one by $d$. We define $u$ as a node corresponding to the initial state of $A_n$.

Consider the following strategy $S_1$ of Player 0. If we come to $t$ from an accepting state, then $S_1$ goes to the initial state by the $c$-colored edge. Otherwise, $S_1$ uses the $d$-colored edge.
Note that $S_1$ is a 2-state strategy -- it just has to remember, whether the last node in a play was some accepting state of $A_n$.

We now show that from any chromatic strategy $S_2$ with $\col(S_2, u)\subseteq\col(S_1, u)$ one can extract a deterministic finite automaton recognizing the language of $A_n$  (below, we denote this language by $L(A_n)$). This means that any such $S_2$ must have $\Omega(3^{n/2})$ states.

 First, notice that, for every word $w$,  Player 1 has a path to $t$ which is colored by $w\#$. Indeed, for any $w$ there is a computation over $w$ which brings either to an accepting or to a rejecting state, from where we can go to $t$ by a $\#$-colored edge. Now, since $A_n$ is self-verifying, we have the following. For $w\in L(A_n)$, there is no $w\#$-colored path to $t$ which goes through a rejecting state (otherwise it would give a computation which rejects $w$). Similarly, for  $w\notin L(A_n)$, there is no $w\#$-colored path to $t$ which goes through an accepting state. This means that if $w\#$ is a prefix of some infinite sequence from $\col(S_1, u)$, then it is followed by $c$ in this sequence if and only if $w\in L(A_n)$. Thus, if $S_2$ is chromatic and  $\col(S_2, u)\subseteq \col(S_1, u)$, then its memory structure, for any $w_1\in L(A_n), w_2\notin L(A_n)$, must come into different states on $w_1$ and on $w_2$. Indeed, otherwise it makes the same move after $w_1\#$ and after $w_2\#$.

\end{proof}

\section{Proof of Theorem \ref{thm:first_upper}}
\label{sec:first_upper}

Let $\mathcal{M} = \langle M, m_{init}, \delta\rangle$ be the memory structure of $S_1$. We have that $|M| = q$.
The set of states of $S_2$ will be the set of functions $f\colon V\to M\cup\{\bot\}$, where $\bot\notin M$. Thus, $S_2$ will be a $(q + 1)^n$-state strategy. The initial state of $S_2$ is the function $f_{init}\colon V\to M\cup\{\bot\}$, 
\[f_{init}(v) = \begin{cases}m_{init} & v\in U,\\ \bot & \mbox{otherwise}\end{cases}\]

We will define $S_2$ in such a way that for any finite path $p$ the following holds. Assume that $p$ is consistent with $S_2$ and $\source(p) \in U$. Let $f\colon V\to M\cup\{\bot\}$ be the state of $S_2$ after $p$. Then we have the following two properties called \emph{soundness} and \emph{completeness}:
\begin{itemize}
\item \emph{(soundness)} for any $v\in V$, if $f(v) = m \neq \bot$, then there exists a finite path $p_1$ with $\source(p_1)\in U, \target(p_1) = v$, such that, first, $p_1$ is consistent with $S_1$, second, $\col(p_1) = \col(p)$, and third,  $\delta(m_{init}, p_1) = m$.
\item \emph{(completeness)} $f(\target(p)) \neq \bot$.
\end{itemize}

Let us first show that for any $S_2$ with these properties we have $\col(S_2, U) \subseteq \col(S_1, U)$. For that it is sufficient to establish the following. Let $P$ be an arbitrary infinite path such that $P$ is consistent with $S_2$ and $\source(P) \in U$. Then there exists an infinite path $P_1$ such that $P_1$ is consistent with $S_1$, $\source(P_1) \in U$ and $\col(P_1) = \col(P)$.

Take an arbitrary $v\in V$. Consider an infinite tree of all finite paths from $v$ that are consistent with $S_1$. Now, delete from this tree all paths that are inconsistent with the coloring of $P$. That is, we delete a path $q$ if $\col(q)\neq\col(p)$, where $p$ is a prefix of $P$ with $|p| = |q|$. Let the resulting tree be $T_v$.

It is sufficient to show that for some $v\in U$, there is an infinite branch in $T_v$. By Kőnig's lemma, we have this as long as there exists $v\in U$ such that $T_v$ is infinite (since we consider only finite arenas, $T_v$ has finite branching for every $v$). To show this, we show that for any $k\in\mathbb{Z}^+$ there exists $v\in U$ such that $T_v$ has a node of depth $k$. Indeed, let $p$ be a prefix of $P$ of length $k$. Since $P$ is consistent with $S_2$, so is $p$. Moreover, $\source(p) = \source(P) \in U$. Let $f\colon V\to M\cup\{\bot\}$ be the state of $S_2$ after reading $p$. By the completeness property, we have $f(\target(p)) = m \neq \bot$. By the soundness property, there exists a finite path $p_1$ with $\source(p_1)\in U$ such that $p_1$ is consistent with $S_1$ and $\col(p_1) = \col(p)$. Observe then that $p_1$ is a depth-$k$ node of $T_{\source(p_1)}$.

\medskip

We now show how to define $S_2$ with properties as above. We first describe the transition function of the memory structure of $S_2$. This memory structure has to be chromatic. So when its transition function receives an edge, it will only use the color of this edge to produce a new state.

Assume the current state of this memory structure is $f\colon V\to M\cup\{\bot\}$, and it receives an edge whose color is $c\in C$. We determine the new state $g\colon V\to M\cup\{\bot\}$ according to the following algorithm. To determine $g(v)$ for $v\in V$, we introduce a notion of a $(f, v, c)$-good edge. An edge $e\in E$ is $(f, v, c)$-good if
\begin{align}
\label{fvc1}
&\mbox{$\target(e) = v$, $\col(e) = c$ and $f(\source(e)) \neq \bot$;}\\
\label{fvc2}
&\mbox{if $\source(e) \in V_0$, then $e = S_1\big(\source(e), f(\source(e))\big)$.}
\end{align}
If no $(f, v, c)$-good edge exists, we set $g(v) = \bot$. Otherwise, we take an arbitrary $(f, v, c)$-good edge $e$ and set $g(v) = \delta\big(f(\source(e)), e\big)$.

We now describe the next-move function of $S_2$. Consider an arbitrary state $f\colon V\to M\cup\{\bot\}$ of $S_2$ and an arbitrary node $v\in V_0$. Define $S_2(v, f)$ as follows. Assume first that $f(v) \neq \bot$. Then set $S_2(v, f) = S_1\big(v, f(v)\big)$. If $f(v) = \bot$, define $S_2(v, f)$ arbitrarily.

\medskip

Definition of $S_2$ is finished. It remains to verify that it satisfies the soundness and the completeness properties. We show this by induction on the length of $p$.

We start with the induction base.
Assume that $p$ is a $0$-length path (then it is automatically consistent with any strategy) and that $\source(p) \in U$. We have to check the soundness and the completeness properties for $p$ and for the initial state $f_{init}$. Let us start with the soundness. If $f_{init}(v) \neq \bot$, then, by definition, $v\in U$ and $f_{init}(v) = m_{init}$. Therefore, we can set $p_1 = \lambda_v$. As for the completeness, we have $f_{init}(\source(p))\neq \bot$ because $\source(p)\in U$.

We now perform the induction step.
Assume that we have verified the soundness and the completeness properties for all paths of length $k$. We extend this to paths of length $k + 1$. Consider  any path $p =p^\prime e^\prime$ of length $k + 1$. Here $e^\prime\in E$ is the last edge of $p$ so that $p^\prime$ is of length $k$. Assume that $p$ is consistent with $S_2$ and $\source(p) \in U$. Then $p^\prime$ is also consistent with $S_2$ and $\source(p^\prime) = \source(p)\in U$. Let $f$ be the state of $S_2$ after $p^\prime$. By the induction hypothesis, we have that the soundness and the completeness hold for $p^\prime$ and $f$. Now, let $g\colon V\to M\cup\{\bot\}$ be the state of $S_2$ after $p$. Alternatively, $g$ is the state into which the memory structure of $S_2$ transits from the state $f$ when it receives $e^\prime$. Let $c = \col(e^\prime)$ be the color of $e^\prime$.

We first show that $p$ and $g$ satisfy the soundness property (see Figure \ref{soundness}).
\begin{figure}[h!]
\centering
  \includegraphics[width=0.8\textwidth]{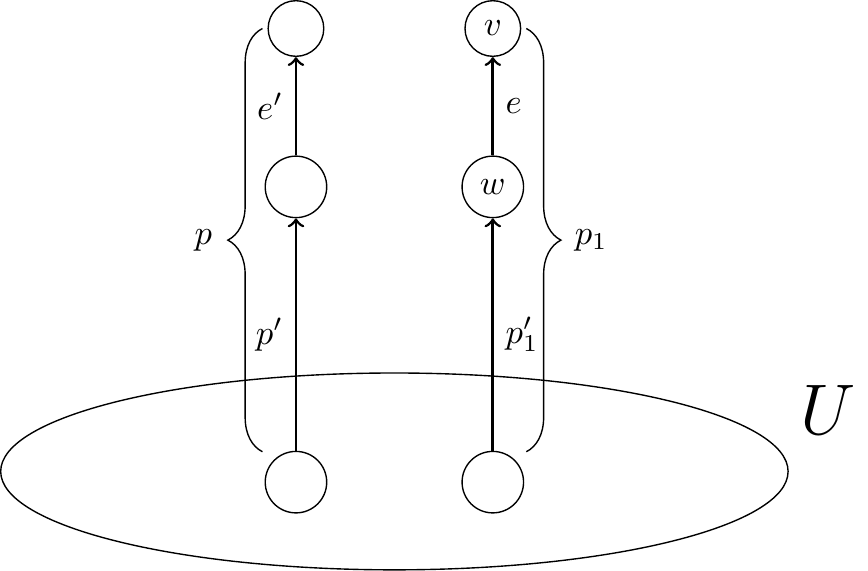}
  \caption{The argument for the soundness.}
\label{soundness}
\end{figure}

 Consider any $v\in V$ such that $g(v)\neq \bot$. There must be an $(f, v, c)$-good edge. Let $e$ be an $(f, v, c)$-good edge which was used to determine $g(v)$. Denote $w = \source(e)$. By \eqref{fvc1}, we have $f(w) \neq \bot$. Hence, by the soundness for $p^\prime$ and $f$, there exists a finite path $p_1^\prime$ with $\source(p_1^\prime)\in U, \target(p_1^\prime) = w$, such that \emph{\textbf{(a)}} $p_1^\prime$ is consistent with $S_1$; \emph{\textbf{(b)}} $\col(p^\prime_1) = \col(p^\prime)$; \emph{\textbf{(c)}} $\delta(m_{init}, p_1^\prime) = f(w)$. Set $p_1 = p_1^\prime e$. Since $\target(p_1^\prime) = w = \source(e)$, we have that $p_1$ is a path. We show that $p_1$ verifies the soundness property for $g(v)$. Obviously, $\source(p_1) = \source(p_1^\prime)\in U$. Since $e$ is $(f, v, c)$-good, we have by \eqref{fvc1} that $\target(e) = v$. Hence, $\target(p_1) = \target(e) = v$. Let us now check that $p_1$ is consistent with $S_1$. This is obvious if $w = \target(p_1^\prime) \in V_1$, because $p_1^\prime$ is consistent with $S_1$. Now, if $w = \target(p_1^\prime)\in V_0$, we have to show that $e = S_1(p_1^\prime)$. Since $f(w) = \delta(m_{init}, p_1^\prime)$ is the state of $S_1$ after $p_1^\prime$, we have $S_1(p_1^\prime) = S_1(w, f(w))$. In turn, since $e$ is $(f, v, c)$-good, by \eqref{fvc2} we have $S_1(w, f(w)) = S_1\big(\source(e), f(\source(e))\big) = e$. It remains to show that $\col(p_1) = \col(p)$ and $\delta(m_{init}, p_1) = g(v)$. Indeed, $\col(p_1) = \col(p_1^\prime e) = \col(p_1^\prime)\col(e) = \col(p^\prime) c = \col(p^\prime) \col(e^\prime) = \col(p^\prime e^\prime) = \col(p)$. Here we use a fact that $\col(e) = c$ due to \eqref{fvc1}. In turn, $\delta(m_{init}, p_1) = \delta(m_{init}, p_1^\prime e) = \delta(\delta(m_{init}, p_1^\prime), e) = \delta(f(w), e)$. It remains to recall that by definition, $g(v) = \delta(f(\source(e)), e) = \delta(f(w), e)$.

Now we show that $p$ and $g$ satisfy the completeness property. In other words, we show that $g(\target(p)) \neq \bot$. By definition, this holds as long as there exists an $(f, \target(p), c)$-good edge. We claim that $e^\prime$, the last edge of $p$, is $(f, \target(p), c)$-good. Let us first verify that $e^\prime$ satisfies \eqref{fvc1}. Obviously, $\target(e^\prime) = \target(p)$. Now, $\col(e^\prime) = c$ by definition. Finally, we have $f(\source(e^\prime)) = f(\target(p^\prime)) \neq \bot$ due to the completeness property for $p^\prime$ and $f$.
Let us now check that $e^\prime$ satisfies \eqref{fvc2}. Assume that $\source(e^\prime) = \target(p^\prime) \in V_0$. Since $p$ is consistent with $S_2$, we have $e^\prime = S_2(p^\prime)$. Now, by definition, $f$ is the state of $S_2$ after $p^\prime$. Therefore, $e^\prime = S_2(p^\prime) =  S_2(\target(p^\prime), f) = S_2(\source(e^\prime), f)$. Again, since the completeness property holds for $p^\prime$ and $f$, we have $f(\source(e^\prime)) = f(\target(p^\prime)) \neq \bot$. Hence, by definition of $S_2$, we have that $e^\prime = S_2(\source(e^\prime), f) = S_1\big(\source(e^\prime), f(\source(e^\prime))\big)$. Thus, \eqref{fvc2} is established for $e^\prime$.

\section{Proof of Theorem \ref{thm:second_upper}}
\label{sec:second_upper}
Let $\mathcal{M} = \langle M, m_{init}, \delta\rangle$ be the memory structure of $S_1$. We have that $|M| = q$. The set of states of $S_2$ will be the set of functions $f\colon V\to V\times M \cup\{\bot\}$, where $\bot\notin V\times M$. Thus, $S_2$ is a $(qn + 1)^n$-state strategy. The initial state of $S_2$ is the function $f_{init}\colon V\to V\times M \cup\{\bot\}$, defined by $f(v) = (v, m_{init})$ for every $v\in V$.

We use the following notation in the proof. Take any $f\colon V\to V\times M\cup\{\bot\}$ and $v\in V$, and assume that $f(v) = (u, m)\neq \bot$. Then we set $f_1(v) = u$ and $f_2(v) = m$. That is, $f_1$ is the projection of $f$ to the first coordinate (its values are nodes of our arena) and $f_2$ is the projection of $f$ to the second coordinate (its values are states of $S_1$). If $f(v) = \bot$, we set $f_1(v) = f_2(v) = \bot$. 

Our goal is to define $S_2$ in a such a way that, for any finite path $p$ which is consistent with $S_2$, and for the state $f$ of $S_2$ after $p$, the following holds:
\begin{itemize}
\item \emph{(soundness)} for any $v\in V$, if $f(v) \neq \bot$, then there exists a finite path $p_1$ from $f_1(v)$ to $v$ such that, first, $p_1$ is consistent with $S_1$, second, $\col(p_1) = \col(p)$, and third,  $\delta(m_{init}, p_1) = f_2(v)$.
\item \emph{(completeness)} $f(\target(p)) \neq \bot$ and $\source(p) \preceq f_2(v)$.
\end{itemize}
It is not hard to see that for any $S_2$ with these properties we have
\[\col(S_2, v) \subseteq \bigcup\limits_{u\in V, v\preceq u} \col(S_1, u).\]
Indeed, to establish this, we have to show that for any infinite path $P$ which is consistent with $S_2$ there exists an infinite path $P_1$ which is consistent with $S_1$ such that $\col(P) = \col(P_1)$ and $\source(P)\preceq\source(P_1)$. For this we define the trees $T_v$, $v\in V$ as in the proof of Theorem \ref{thm:first_upper}. By Kőnig's lemma, it is sufficient to show that $T_u$ is infinite for some $u$ with  $\source(P) \preceq u$.  We take an arbitrary finite prefix $p$ of $P$. Since $p$ is consistent with $S_2$, from the completeness we get that $f(\target(p))\neq \bot$. By applying the soundness to the node $\target(p)$, we get a finite path $p_1$ from $f_1(\target(p))$ to $\target(p)$ such that, first, $p_1$ is consistent with $S_1$, and second, $\col(p_1) = \col(p)$. Hence, $p_1$ is a node of $T_{f_1(\target(p))}$. Moreover, by the completeness we have $\source(P) = \source(p) \preceq f_1(\target(p))$.
Thus, for some $u$ with $\source(P) \preceq u$ there is a node of depth $|p|$ in $T_u$. It remains to note that $|p|$ can be arbitrarily large.

\medskip

We now explain how to define $S_2$ in a way which guaranties the soundness and the completeness properties.
We start with the transition function of $S_2$. Assume that the current state of $S_2$ is $f\colon V\to V\times M \cup \{\bot\}$, and then it receives an edge whose color is $c$. We define the new state $g\colon V\to V\times M \cup \{\bot\}$ as follows (we stress that $S_2$ has to be chromatic, so $g$ will be a function of $f$ and $c$). Take any $v\in V$ for which we want to determine $g(v)$. Note that $f_2\colon V\to M\cup\{\bot\}$. If there is no $(f_2, v, c)$-good edge, in a sense of (\ref{fvc1}--\ref{fvc2}), then we set $g(v) = \bot$. Otherwise, we take an $(f_2, v, c)$-good edge $e$, maximizing $f_1(\source(e))$ w.r.t.~$\preceq$, and set $g_1(v) = f_1(\source(e)), g_2(v) = \delta(f_2(\source(e)), e)$.

We now define the next-move function of $S_2$. Let $f\colon V\to V\times M\cup\{\bot\}$ be a state and $v\in V_0$ be a node of Player $0$. If $f(v) \neq \bot$, we set $S_2(v, f) = S_1(v, f_2(v))$. Otherwise, we define $S_2(v, f)$ arbitrarily.

\medskip

It remains to establish the soundness and the completeness properties for all finite paths $p$ that are consistent with $S_2$. As before, we do so by induction on $|p|$.

We start with the induction base. Assume that $|p| = 0$. The initial state of $S_2$ is the function $f_{init}$. Recall that we have $f_{init}(v) = (v, m_{init})$ for every $v\in V$. So, to establish the soundness, we can set $p_1 = \lambda_v$ for every $v\in V$. For the completeness, observe that $f_{init}(\target(p)) = (\target(p), m_{init}) \neq \bot$ and, obviously, $\source(p) \preceq \target(p)$ (just because $p$ is a $0$-length path so that $\source(p) = \target(p)$).

Let us now perform the induction step. Assume that our claim is proved for all $p$ of length up to $k$. Take any $p = p^\prime e^\prime$ of length $k + 1$ which is consistent with $S_2$. Here $e^\prime$ is the last edge of $p$. Then $p^\prime$ is consistent with $S_2$ and has length $k$. Hence, we have the induction hypothesis for $p^\prime$ and for a function $f\colon V\to V\times M\cup\{\bot\}$ which is the state of $S_2$ after $p^\prime$. Next, let the state of $S_2$ after $p$ be $g\colon V\to V\times M \cup\{\bot\}$. Note that $g$ is the value of the transition function of $S_2$ on $f$ and $c = \col(e^\prime) \in C$.

To check the soundness for $p$ and $g$, one can use exactly the same argument as in Theorem \ref{thm:first_upper} for $f_2$ and $g_2$. That is, for any $v\in V$ with $g(v) \neq \bot$, we consider an $(f_2, v, c)$-good edge $e$ which was used to define $g(v)$. By \eqref{fvc1}, we have $f_2(\source(e)) \neq \bot \implies f(\source(e)) \neq \bot$. Then, using the induction hypothesis for $p^\prime$, we take a path $p_1^\prime$ establishing the soundness for $f$ and $p^\prime$ at $\source(e)$. Finally, we define $p_1 = p_1^\prime e$ and show that $p_1$ establishes the soundness for $g$ and $p$ at $v$. Obviously, $p_1$ is a path to $v$. By the same routine check as in the proof of Theorem \ref{thm:first_upper}, we have that, first, $p_1$ is consistent with $S_1$, second, $\col(p_1) = \col(p)$, and third, $\delta(m_{init}, p_1) = g_2(v)$. 
The only thing we have to additionally check is that $p_1$ starts in $g_1(v)$. Indeed, by definition, $g_1(v) = f_1(\source(e))$. Note that $p_1^\prime$ is a prefix of $p_1$, so these paths have the same starting node. In turn, since $p^\prime_1$ establishes the soundness for $f$ and $p^\prime$ at $\source(e)$, the starting node of $p_1^\prime$ must be $f_1(\source(e)) = g_1(v)$, as required.

We now check the completeness property for $p$ and $g$. It is sufficient to show the existence of an $(f_2, \target(p), c)$-good edge $e$ with $\source(p) \preceq f_1(\source(e))$.  Indeed, $g(\target(p)) \neq \bot$ if and only if $(f_2, \target(p), c)$-good edges exist, and $g_1(\target(p))$ is the maximum of $f_1(\source(e))$ w.r.t.~$\preceq$ over such edges.

We claim that $e^\prime$, the last edge of $p$, satisfies these conditions. To show this, recall that by the induction hypothesis we have the completeness  for $p^\prime$ and $f$. Let us first demonstrate that $e^\prime$ satisfies \eqref{fvc1} for $f_2$, $v = \target(p)$ and $c$. Indeed,  $\target(e^\prime) = \target(p)$ because $e^\prime$ is the last edge of $p$, $\col(e^\prime) = c$ by definition of $c$, and $f_2(\source(e^\prime)) = f_2(\target(p^\prime)) \neq \bot$ by the completeness for $p^\prime$ and $f$. Let us now verify \eqref{fvc2}. Assume that $\source(e^\prime) \in V_0$. Then, since $p = p^\prime e^\prime$ is consistent with $S_2$, we have $e^\prime = S_2(p^\prime)$. The state of $S_2$ after $p^\prime$ is $f$, so $S_2(p^\prime) = S_2(\target(p^\prime), f) = S_2(\source(e^\prime), f)$. Note that $f(\source(e^\prime)) = f(\target(p^\prime)) \neq \bot$ by the completeness for $p^\prime$ and $f$. Hence, by definition of $S_2$, we have  $S_2(\source(e^\prime), f) = S_1(\source(e^\prime), f_2(\source(e^\prime))$, and, thus, $e^\prime$
satisfies \eqref{fvc2}. Finally, we have to show that $\source(p) \preceq f_1(\source(e))$. This is because, by the completeness for $p^\prime$ and $f$, we have $\source(p^\prime) \preceq f_1(\target(p^\prime))$. It remains to note that $\source(p) = \source(p^\prime)$ and $\source(e^\prime) = \target(p^\prime)$ -- recall that $e^\prime$ is the last edge of $p$ and $p^\prime$ is the part of $p$ which precedes $e^\prime$.

\section{Proof of Theorem \ref{thm:lower_bound}}
\label{sec:lower_bound}
Let $\mathcal{A}$ be as on Figure \ref{arena}. 

\begin{figure}[h!]
\centering
  \includegraphics[width=\textwidth]{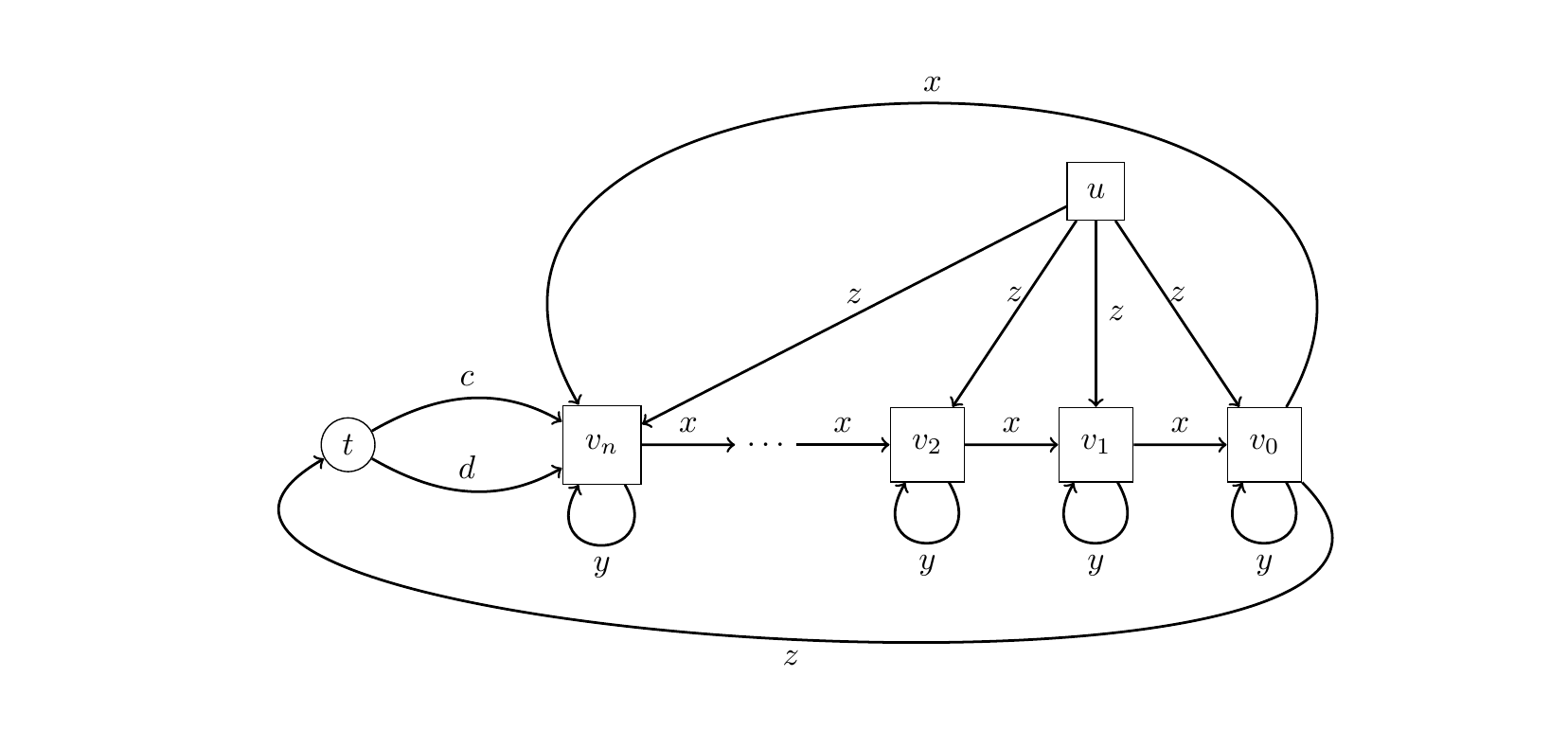}
  \caption{Arena $\mathcal{A}$. The set of colors is $C = \{x, y, z, c, d\}$. The partition of the nodes between the players is given by $V_0 = \{t\}, V_1 = \{u, v_0, v_1, \ldots, v_n\}$.}
\label{arena}
\end{figure}

We define $S_1$ as follows. Its memory structure maintains a number $\coun\in\{0, 1, \ldots, q - 1\}$. Initially, $\coun = 0$. When the memory structure of $S_1$ passes through any $y$-colored edge, it increments $\coun$ by $1$ modulo $q$. In turn, when we go from $v_0$ to $v_n$, it sets $\coun = 0$. In all the other cases, the value of $\coun$ does not change. It remains to define how $S_1$ acts at $t$ (this is the only node from where Player 0 is the one to move). There are two edges from $t$, both go to $v_0$, but one is $c$-colored and the other is $d$-colored. If $\coun = 0$, then $S_1$ uses the $c$-colored edge. If $\coun\neq 0$, then $S_1$ uses the $d$-colored edge.

For brevity, if $p$ is a finite path in $\mathcal{A}$, then by $\coun(p)$ we denote the value of $\coun$ after $p$.

We need the following definition and the following lemma about paths in the arena $\mathcal{A}$.
\begin{definition} Define a function $f\colon\{x, y\}^*\to\{0, 1, \ldots, q - 1\}$ as follows. Take any $w\in\{x, y\}^*$. To define $f(w)$, first define a word $w^\prime \in\{x, y\}^*$. Namely, if $w$ has at most $n$ occurrences of $x$, then set $w^\prime = w$. Otherwise, take the $(n+1)$st occurrence of $x$ from the right, erase it and everything to its left, and let the remaining word be $w^\prime$. Finally, let $f(w)$ be the number of $y$'s in $w^\prime$ modulo $q$.
\end{definition}

\begin{lemma}
\label{lemma:paths}
For any $w\in\{x, y\}^*$ the following holds:
\begin{description}
\item[\textup{(a)}] there exists a finite path $p$ with $\source(p) = u$, $\target(p) = t$ and $\col(p) = zwz$;
\item[\textup{(b)}] for any finite path $p$, if $\source(p) = u$ and $\col(p) = zwz$, then $\target(p) = t$ and $\coun(p) = f(w)$.
\end{description}
\end{lemma}
\begin{proof}
 Let $X$ be the number of occurrences of $x$ in $w$. Set $i$ to be the remainder of $X$ when divided by $n + 1$.

We start by showing \textbf{\textup{(a)}}. To construct $p$, we first go from $u$ to $v_i$. Then we start reading letters of $w$ one by one from left to right. Every time we read a new letter, we move from our current location via some edge colored by this letter. It remains to show that after reading the whole $w$ we end up in $v_0$, which has an out-going $z$-colored edge to $t$.
 Indeed, if we forget about $y$'s, then we are just rotating counterclockwise along the cycle $v_n\to\ldots\to v_1\to v_0\to v_n$. The length of this cycle is $n + 1$, and the distance from $v_i$ to $v_0$, measured counterclockwise, is $i$. Thus, since in $w$ there are $X\equiv i \pmod{n+1}$ occurrences of $x$, we end up in $v_0$.

We now show \textbf{\textup{(b)}}. Consider any finite path $p$ with $\source(p) = u$ and $\col(p) = zwz$. Observe that once we left $u$, it is impossible to come back to it again. Therefore, since the last edge of $p$ is $z$-colored, this edge must be from $v_0$ to $t$. Hence, $\target(p) = t$.

It remains to show that $\coun(p) = f(w)$. Assume first that $p$ never goes from $v_0$ to $v_n$. Then $\coun(p)$ is the number of $y$'s modulo $q$ in $w$, because $\col(p) = zwz$. Thus, to show that $\coun(p) = f(w)$ in this case, it is enough  to show $X\le n$. Indeed, the first edge of $p$ is from $u$ to $v_j$, for some $j\in\{0, 1, \ldots, n\}$. Then it makes $X$ steps along the cycle $v_n\to\ldots\to v_1\to v_0\to v_n$. If $X$ were at least $n + 1$, then $p$ had to go from $v_0$ to $v_n$ at least once, contradiction.

Now, assume that $p$ contains edges from $v_0$ to $v_n$. By definition, $\coun(p)$ equals the number of $y$-colored edges in $p$ modulo $q$ after the last time $p$ went from $v_0$ to $v_n$. To show that $\coun(p) = f(w)$, we have to show that the number of $x$-colored edges in $p$ after the last time $p$ went from $v_0$ to $v_n$ is $n$ (then the last edge from $v_0$ to $v_n$ in $p$ corresponds to the $(n+1)$st occurrence of $x$ in $w$ from the right). Indeed, as we discussed above, the last edge of $p$ must be from $v_0$ to $t$. Obviously, if we go from $v_n$ to $v_0$ without going to $v_n$ again after this, then the number of times we pass an $x$-colored edge is exactly $n$.

\end{proof}

This gives the following fact about the set $\col(S_1, u)$.
\begin{corollary}
\label{cor:prefix}
 For any $w\in\{x, y\}^*$ the following holds. If $zwzc$ is a prefix of some  sequence from $\col(S_1, u)$, then $f(w) = 0$. In turn, if $zwzd$ is a prefix of some sequence from $\col(S_1, u)$, then $f(w) \neq 0$.
\end{corollary}
\begin{proof}
Fix $h\in \{c, d\}$. Take any $w\in\{x, y\}^*$ such that $zwzh$ is a prefix of some sequence of $\col(S_1, u)$. We show that $h = c \iff f(w) = 0$.

By definition of $\col(S_1, u)$, there exists a finite path $p$ with $\source(p) = u, \col(p) = zwzh$ which is consistent with $S_1$. Let $p_1$ be the part of $p$ which precedes its last edge. Since $\source(p_1) = u$ and $\col(p_1) = zwz$, we have by the item \textbf{\textup{(b)}} of Lemma \ref{lemma:paths} that $\target(p_1) = t$ and $\coun(p_1) = f(w)$. 

The node $t$ is controlled by Player $0$. Hence, since $p$ is consistent with $S_1$, the last edge of $p$ must be equal to $S_1(p_1)$. The color of $S_1(p_1)$ is $h$. In turn, $\coun(p_1)$ is the state of $S_1$ after $p_1$. Therefore, by definition of $S_1$, the color of $S_1(p_1)$ is $c$ if and only if $\coun(p_1) = f(w) = 0$. The lemma is proved. 
\end{proof}

Consider now any chromatic $Q$-state strategy $S_2$ such that $\col(S_2, u) \subseteq \col(S_1, u)$. Let $\mathcal{M} = \langle M, m_{init}, \sigma\colon M\times\{x, y, z, c, d\}\to M\rangle$ be its (chromatic) memory structure. To show that $Q\ge q^n$, in Lemma \ref{lemma:different} we provide $q^n$ words from $\{x, y, z, c, d\}^*$ such that $\sigma(m_{init}, \cdot)$ must  take different values on these words.

\begin{definition}
Let $g\colon\{0, 1, \ldots, q -1\}^n\to \{x, y\}^*$ be the following function:
\[g\colon (i_1, i_2, \ldots, i_n) \mapsto  x y^{i_1} xy^{i_2}\ldots x y^{i_n}.\]
\end{definition}

\begin{lemma}
\label{lemma:different}
 For any $\kappa_1, \kappa_2\in\{0, 1, \ldots, q - 1\}^n$ such that $\kappa_1\neq \kappa_2$ we have $\sigma(m_{init}, zg(\kappa_1)) \neq \sigma(m_{init}, zg(\kappa_2))$.
\end{lemma}

To establish Lemma \ref{lemma:different}, we first need the following lemma.

\begin{lemma}
\label{lemma:word}
 For any $\kappa_1, \kappa_2\in\{0, 1, \ldots, q - 1\}^n$ such that $\kappa_1\neq \kappa_2$ there exists a word $w\in \{x, y\}^*$ such that $f(g(\kappa_1) w) = 0$ and $f(g(\kappa_2)w)\neq 0$.
\end{lemma}
\begin{proof}
Assume that $\kappa_1 = (i_1, i_2, \ldots, i_n)$ and $\kappa_2 = (j_1, \ldots, j_n)$. Take the largest $k\in\{1,2, \ldots, n\}$ such that $i_k \neq j_k$. Let $r\in\{0, 1, \ldots, q - 1\}$ be such that $i_k + i_{k + 1} + \ldots + i_n + r\equiv 0 \pmod{q}$. Define $w = x^ky^r$. Thus, 
\begin{align*}
g(\kappa_1) w &=  x y^{i_1} xy^{i_2}\ldots  xy^{i_k} \ldots x y^{i_n} x^k y^r ,\\
g(\kappa_2) w &=  x y^{j_1} xy^{j_2}\ldots  xy^{j_k} \ldots x y^{j_n} x^k y^r.
\end{align*}
Observe that the $(n+1)$st occurrence of $x$ in $g(\kappa_1) w$ is one before $y^{i_k}$. Similarly, the $(n+1)$st occurrence of $x$ in $g(\kappa_2) w$ is one before $y^{j_k}$.
Hence, by definition of $f$, we have:
\begin{align*}
f(g(\kappa_1) w) &\equiv  i_k + i_{k + 1} + \ldots + i_n + r\pmod{q},\\
f(g(\kappa_2) w) &\equiv j_k + j_{k + 1} + \ldots + j_n + r\pmod{q}.
\end{align*}
By definition of $r$, we have $f(g(\kappa_1) w) = 0$. In turn, by definition of $k$, we have $i_k \neq j_k$ and $i_{k + 1} = j_{k + 1}, \ldots, i_n = j_n$. The numbers $i_k, j_k$ are different elements of $\{0, 1, \ldots, q - 1\}$, which means that their difference is not divisible by $q$. Hence,  $f(g(\kappa_2) w) \neq f(g(\kappa_1) w) = 0$.
\end{proof}

To conclude the proof of the theorem, it remains to derive Lemma \ref{lemma:different} from Lemma \ref{lemma:word}.
Assume for contradiction that 
\begin{equation}
\label{trans_eq}
\sigma(m_{init}, zg(\kappa_1)) = \sigma(m_{init}, zg(\kappa_2))
\end{equation}
 for some $\kappa_1, \kappa_2\in\{0, 1, \ldots, q - 1\}^n$, $\kappa_1\neq \kappa_2$. By Lemma \ref{lemma:word} there exists $w\in \{x, y\}^*$ such that
\begin{equation} 
\label{trans_ineq}
f(g(\kappa_1) w) = 0,\qquad f(g(\kappa_2)w)\neq 0.
\end{equation}
 By the item \textbf{\textup{(a)}} of Lemma \ref{lemma:paths} there exist two finite paths $p_1$ and $p_2$ such that
\begin{align*}
\source(p_1) &= \source(p_2) = u, \\
\target(p_1) &= \target(p_2) = t, \\
\col(p_1) &= zg(\kappa_1) wz, \qquad \col(p_2) = zg(\kappa_2) wz.
\end{align*}
The paths $p_1$ and $p_2$ do not have $c,d$-colored edges. That is, they do not have edges that start at $t$. This means that these paths are consistent with $S_2$. We claim that $S_2(p_1) = S_2(p_2)$. Indeed, since $S_2$ is chromatic, the value of $S_2(p)$ for a finite path $p$ with $\target(p) = t$ is completely determined by $\sigma(m_{init}, \col(p))$. Now, by \eqref{trans_eq} we have that:
\[
\sigma(m_{init}, \col(p_1)) = \sigma(m_{init}, zg(\kappa_1) wz) = \sigma(m_{init}, zg(\kappa_2) wz) = \sigma(m_{init}, \col(p_2)).
\]
So let $e = S_2(p_1) = S_2(p_2)$. The paths $p_1 e$ and $p_2 e$ are both consistent with $S_2$. Since $\col(S_2, u) \subseteq \col(S_1, u)$, we have that $\col(p_1 e)$ is a prefix of some sequence from $\col(S_1, u)$, and so is $\col(p_2 e)$. This gives a contradiction with Corollary \ref{cor:prefix}. Indeed, assume first that $\col(e) = c$. Then $\col(p_2 e) = zg(\kappa_2) wz c$ is a prefix of some sequence from $\col(S_1, u)$, but $f(g(\kappa_2) w) \neq 0$ by \eqref{trans_ineq}, contradiction. Similarly, if $\col(e) = d$, then $\col(p_1 e) = zg(\kappa_1) wz d$ is a prefix of some sequence from $\col(S_1, u)$,  but $f(g(\kappa_1) w) = 0$ by \eqref{trans_ineq}, contradiction.

\end{document}